\documentclass[orivec]{llncs}

\usepackage[T1]{fontenc}
\usepackage{hyperref}
\usepackage{amsmath}
\usepackage{amssymb}
\usepackage{graphicx}
\graphicspath{{fig/}}
\usepackage[usenames,dvipsnames]{color} 
\usepackage{mdwmath}
\usepackage{verbatim}
\usepackage{amssymb}
\usepackage{tikz}
\usepackage{stmaryrd}

\newcommand{\server}{\ensuremath{\texttt{S}}}

\newcommand{\OR}{\ensuremath{\mathtt{OR}}}
\newcommand{\AND}{\ensuremath{\mathtt{AND}}}
\newcommand{\SAND}{\ensuremath{\mathtt{SAND}}}
\newcommand{\OAND}{\ensuremath{\mathtt{OAND}}}

\newcommand{\basicact}{\ensuremath{\mathbb{B}}}
\newcommand{\atree}{\ensuremath{\mathbb{T}}}
\newcommand{\sandtree}{\ensuremath{\mathbb{T}_{\SAND}}}
\newcommand{\sandtreevar}{\ensuremath{\sandtree^{\var}}}
\newcommand{\atreevar}{\ensuremath{\atree^{\var}}}

\newcommand{\power}[1]{\ensuremath{\mathcal{P}(#1)}}

\newcommand{\msl}{\{\hspace*{-0.1cm}|}
\newcommand{\msr}{|\hspace*{-0.1cm}\}}

\newcommand{\semantic}[1]{[\hspace{-0.05cm}[#1]\hspace{-0.05cm}]}
\newcommand{\sem}[1]{[\hspace{-0.05cm}[#1]\hspace{-0.05cm}]_{\mathcal{S\!P}}}
\newcommand{\msem}[1]{[\hspace{-0.05cm}[#1]\hspace{-0.05cm}]_{\mathcal{M}}}

\newcommand{\powerset}{\mathcal{P}}
\newcommand{\seqop}{\cdot}
\newcommand{\parop}{\parallel}

\newcommand{\spset}{\ensuremath{\mathbb{G}_{\mathcal{S\!P}}}}

\newcommand{\var}{\mathbb{V}} 
\newcommand{\ESP}{E_{\mathcal{S}\mathcal{P}}}
\newcommand{\ESM}{E_{\mathcal{M}}}
\newcommand{\EADT}{E_{ADT}}
\newcommand{\thetrs}{R_{\mathcal{S}\mathcal{P}}}

\newcommand{\mi}[1]{\mathit{#1}} 
\newcommand{\hy}[1]{\text{-}} 
\newcommand{\becomeroot}{\ensuremath{\textit{become root}}}
\newcommand{\noauth}{\ensuremath{\textit{no-auth}}}
\newcommand{\auth}{\ensuremath{\textit{auth}}}

\newcommand{\ftprhostsx}{\ensuremath{\textit{ftp}}}
\newcommand{\rshx}{\ensuremath{\textit{rsh}}}
\newcommand{\localbofx}{\ensuremath{\textit{lobf}}}
\newcommand{\sshbofx}{\ensuremath{\textit{ssh}}}
\newcommand{\rsarefbofx}{\ensuremath{\textit{rsa}}}

\newcommand{\fullattr}{\ensuremath{A_\alpha}}
\newcommand{\attr}{\ensuremath{\alpha}}
\newcommand{\attrdomain}{\ensuremath{D}}
\newcommand{\attrval}{\ensuremath{V}}

\newcommand{\attrand}{\ensuremath{\vartriangle}} 
\newcommand{\attror}{\ensuremath{\triangledown}} 
\newcommand{\attrsand}{\ensuremath{\Diamond}} 

\newcommand{\basicassign}{\ensuremath{\beta}}

\newcommand{\nfterm}{\ensuremath{\mi{N}}}
\newcommand{\asterm}{\ensuremath{\mi{C}}}
\newcommand{\sandterm}{\ensuremath{\mi{S}}}
\newcommand{\andterm}{\ensuremath{\mi{A}}}
\newcommand{\nfterms}{\ensuremath{\mathbb{T}_{\nfterm}}}
\newcommand{\asterms}{\ensuremath{\mathbb{T}_{\asterm}}}
\newcommand{\andterms}{\ensuremath{\mathbb{T}_{\andterm}}}
\newcommand{\sandterms}{\ensuremath{\mathbb{T}_{\sandterm}}}

\newcommand{\Xvec}{\overline{X}}
\newcommand{\Yvec}{\overline{Y}}
\newcommand{\Zvec}{\overline{Z}}

\newcommand{\norm}[1]{\lvert #1 \rvert}

\newcommand{\eg}{e.g.}
\newcommand{\ie}{i.e.}

\begin{document}

\title{{Attack Trees with Sequential Conjunction}\thanks{This is 
an extended version of~\cite{IFIP-SEC'15}.}}
\author{Ravi Jhawar\inst{1}, Barbara Kordy\inst{2}, Sjouke Mauw\inst{1}, 
Sa\v{s}a Radomirovi\'c\inst{3},\\ Rolando Trujillo-Rasua\inst{1}}
\institute{University of Luxembourg, SnT, Luxembourg\\
\and INSA Rennes, IRISA, France\\
\and Inst.~of Information Security, Dept.~of Computer Science,
ETH Z\"{u}rich, Switzerland
}

\maketitle

\begin{abstract}
We provide the first formal foundation of 
\SAND~attack trees which are a popular extension of the
well-known attack trees. 
The \SAND~attack tree formalism increases the expressivity of 
attack trees by introducing the sequential conjunctive operator $\SAND$. This
operator enables the modeling of ordered events.

We give a semantics to \SAND~attack trees by interpreting them as 
sets of 
series-parallel graphs and propose a complete axiomatization of this semantics. 
We define normal forms for \SAND~attack trees and a term rewriting system
which allows identification of semantically equivalent 
trees. 
Finally, we formalize how to
quantitatively analyze \SAND~attack trees using attributes.

\end{abstract}

\section{Introduction}
\label{sec:introduction}

Attack trees allow for an effective security analysis by systematically organizing the different 
ways in which a system can be attacked into a tree. 
The root node of an attack tree 
represents the {\em attacker's goal} and the children of a given node represent 
its refinement into {\em sub-goals}. A refinement is typically either disjunctive 
(denoted by $\OR$) or conjunctive (denoted by $\AND$). 
The 
leaves of an attack tree represent the attacker's actions and are called {\em basic 
actions}. 

Since their inception by Schneier~\cite{Schn}, attack trees
have quickly become a popular modeling tool for security 
analysts.
However, the limitations of this formalism,
in particular with respect to expressing the order in which the
various attack steps are executed, have been recognized by
many authors (see \eg,~\cite{KoPiSc_CSR14}).
In practice, modeling of security scenarios often requires constructs
where conditions on the execution order of the attack components can
be clearly specified. This is for instance the case when the time or
(conditional) probability 
of an attack is considered, as
in~\cite{Arnold-POST14,LL2011}. Consequently, several studies have
extended attack trees informally with sequential conjunctive
refinements. Such extensions have resulted in improved modeling and
analyses (e.g., \cite{LL2011,PB2010,JW2009}) and
software tools, e.g., ATSyRA~\cite{pinchinat:hal-01064645}.

Even though the sequential conjunctive refinement, 
that we denote by \SAND, 
is well understood at a conceptual level 
and even applied to real world scenarios~\cite{pinchinat:hal-01064645}, 
none of the existing solutions have provided 
a rigorous mathematical formalization of attack trees with \SAND. 
Indeed, the extensions found in the literature 
are rather diverse in terms of application domain, interpretation, and 
formality. Thereby, it is infeasible
 to answer fundamental questions such as: 
What is the precise expressibility of \SAND~attack trees?
When do two such trees 
represent the same security 
scenario? Or what type of attributes can be synthesized 
on \SAND~attack trees in the standard bottom-up way?
These questions can only be precisely answered if \SAND~attack trees are 
provided with a formal, general, and explicit interpretation, that is to say, 
if \SAND~attack trees are given a formal foundation. 

\noindent{\em Contributions:} 
In this article we formalize 
the meaning of a $\SAND$ attack tree by defining its semantics. Our 
semantics 
is based on series-parallel (SP) graphs, which is a well-studied branch of 
graph theory. We provide a complete 
axiomatization 
for the SP semantics and show that the 
SP semantics for $\SAND$ attack trees are a conservative extension of the 
multiset semantics for standard attack trees~\cite{MaOo} (i.e., our 
extension 
does not 
introduce unexpected equivalences w.r.t.\ the multiset semantics). 
To do so, 
we define a term 
rewriting system that is terminating and confluent
and obtain normal forms for $\SAND$ attack trees.
As a consequence, we achieve the rather surprising result that the
domains of $\SAND$ attack trees
and sets of SP graphs are isomorphic. 
We also extend the notion of attributes 
for $\SAND$ attack trees which enable 
the quantitative analysis of attack scenarios using the standard bottom-up 
evaluation algorithm. 

{
One of the goals of our work is to provide a level of abstraction that
encompasses most of the existing approaches from literature. For example,
operators, 
such as the priority-based and the
time-based connectors~\cite{WaWhPhPa}, are indeed captured by the
$\SAND$ operator defined in this article. Moreover, other published
semantics, such as those based on cumulative distribution
functions~\cite{Arnold-POST14}, conditional
probabilities~\cite{WaWhPhPa}, or boolean algebra~\cite{Khan}, can be
expressed as an attribute in our formalism. Last but not least, even
though we make the distinction between $\AND$ and $\SAND$ refinements
explicit, our semantics satisfies backward compatibility with the
well-known multiset semantics of attack trees~\cite{MaOo}.
This stresses the, much needed, unifying character of our approach.
}

\noindent{\em Organization:} Section~\ref{sec:related-work} summarizes the 
related work and puts our work in context. 
Section~\ref{sec:sand-atree} provides a formal definition of \SAND~attack trees
and its semantics using series-parallel graphs. Section~\ref{sec:axioms}
defines a complete set of axioms for \SAND~attack trees and presents 
a term rewriting system which allows identification of
semantically equivalent \SAND~attack trees. Section~\ref{sec:attributes}
outlines an approach to quantitatively analyze \SAND~attack trees using 
attributes. Finally, Section~\ref{sec:conclusions} concludes with an outlook 
on future work.

\section{Related Work and Motivation}
\label{sec:related-work}

{}
Several extensions of attack trees with temporal or 
causal dependencies between 
attack steps have been proposed.
We observe that there are three different approaches to achieve this
goal.
The first approach is to use standard attack trees with the added
assumption that the children of an \AND\ node are sequentially
ordered. This approach is mostly applied to the design
of algorithms or tools for the analysis of attack trees under the
assumption of ordered events. 

The second approach is to introduce a mechanism for ordering
events in an attack tree, for instance by adding a
new type of edge to express causality or conditionality. In its most
general case, any partial order on the events in an attack tree can be
specified.
The third approach consists of the introduction of a new type of
node for sequencing. Most extensions fall in this category.
This approach is used by authors who require their formalism to be
backward compatible, or who need standard, as well as ordered
conjunction. We discuss for each of these approaches 
the most relevant papers with respect to the present article.
{
That is, we only consider approaches that still have the main
characteristics of attack trees, being the presence of \AND\ and \OR\
nodes and the interpretation of the edges as a refinement relation.
Thus, we consider approaches such as attack
graphs~\cite{SHJLW2002,OBM2006} and Bayesian
networks~\cite{QL2004,JW2007,LM2005} as out of scope
for this paper.
}

\subsubsection{Approaches with a sequential interpretation of \AND.}
In their work on \emph{Bayesian networks for security}, Qin and
Lee define a transformation from attack trees to Bayesian
networks~\cite{QiLe}. They state that ``there always exists
an implicit dependent and sequential relationship between \AND\ nodes
in an attack tree.'' Most literature on attack trees seem to
contradict this statement, implying that there is a need to explicitly
identify such sequential relationships.

J\"urgenson and Willemson developed an algorithm to calculate the
\emph{expected outcome} of an attack tree~\cite{JW2009}. The goal of
the algorithm is to determine a permutation of leaves for which the
optimal expected outcome for an attacker can be achieved. In essence,
their input is an attack tree where an \AND\ node represents all
possible sequences of its children.
A peculiarity of their interpretation is that multiple occurrences of
the same node are considered only once, implying that the execution of
twice the same action cannot be expressed. 

\subsubsection{Approaches introducing a general order.}
{
Peine, Jawurek, and Mandel 
introduce \emph{security goal indicator trees}~\cite{PJM2008}
in which nodes can be related by a notion of \emph{conditional
dependency} and Boolean connectors. 
The authors, 
however, do not formally specify the syntax and 
semantics of the model. 
} 
A more general approach is proposed by Pi\`etre-Cambac\'ed\`es and 
Bouissou 
\cite{PB2010}, who apply 
\emph{Boolean logic driven Markov
processes} to security modeling. Their formalism does not introduce new
gates, but a (trigger-)relation on the nodes of the attack tree.
Although triggers 
can express a more general sequential relation than the $\SAND$
operator, they lack the readability of standard attack tree operators. 

\emph{Vulnerability cause graphs}~\cite{ArBySh,ByArShDu} combine properties of attack 
trees ($\AND$ and $\OR$ nodes) and attack graphs (edges express order 
rather 
than refinement). The interaction between the $\AND$ nodes and the order 
relation is defined through a graph transformation called \emph{conversion of 
conjunctions}, which ignores the order between nodes. This 
discrepancy could be solved by considering distinct
conjunctive and sequential conjunctive nodes, as we do in this paper.

\subsubsection{Approaches introducing sequential \AND.}
As noted by Arnold et 
al.~\cite{Arnold-POST14}, the analysis of time-dependent 
attacks requires attack trees to be extended with a sequential operator. This 
is
accomplished by defining sequential nodes as conjunctive nodes with a
notion of progress of time.
The authors define a formal semantics for this extension based on cumulative
distribution functions (CDFs), where a CDF denotes the probability that a
successful attack occurs within time $t$. The main difference with our work is 
that their approach is based on
an explicit notion of time, while we have a more abstract approach
based on causality. In their semantics, the meaning of an extended
attack tree is a CDF, in which the relation to the individual basic
attacks is not explicit anymore. In contrast, in our semantics
the individual basic attacks and their causal ordering remain visible.
As such, our semantics can be considered more abstract,
and indeed, we can formulate their semantics as an \emph{attribute} in our
approach.

\emph{Enhanced attack trees}~\cite{AY2006} (EATs) distinguish between $\OR$, $\AND$ 
and $\OAND$ (Ordered $\AND$). Similarly to the approach of Arnold et
al.~\cite{Arnold-POST14}, ordered $\AND$ nodes are used to express 
temporal dependencies between attack components. The authors evaluate EATs by 
transforming them into tree automata. Intermediate states in the automaton 
support the task of reporting partial attacks. However, because every 
intermediate node of the tree corresponds to a state in the tree automaton, 
their approach does not scale well. This problem can be addressed by 
considering 
the normal form of attack trees, as proposed in this article. 

Not every extension of attack trees with $\SAND$ refinements concerns 
time-dependent attack scenarios; some aim at supporting risk 
analyses with conditional
probabilities. For that purpose, Wen-Ping and Wei-Min 
introduce \emph{improved attack trees}~\cite{LL2011}. The concepts, 
however, are described at an intuitive level only. 

\emph{Unified parameterizable attack trees}~\cite{WaWhPhPa} unify different extensions 
of attack trees (structural, computational, and hybrid). The authors 
consider two types of ordered $\AND$ connectors: priority-based connectors
and time-based connectors. 
The children of the former are ordered from
highest to lowest priority, whereas the children of the latter are
ordered temporally. 
Our formalism gives a single interpretation to the 
$\SAND$ operator, yet it can capture both connectors. 

{
Due to obvious similarities, we also review approaches that 
introduce the $\SAND$ operator in fault 
trees. For example, Brooke and Paige include
five fault tree gates: $\AND$, $\OR$, priority $\AND$, exclusive $\OR$,
and an inhibit gate~\cite{BP2003}. The authors do not
discuss the semantics of their model for security, though. Another fault tree 
based approach is discussed by
Khand~\cite{Khan}, who proposes to extend attack trees with a set of gates
from dynamical fault tree modeling that overlaps with the gates used
by Brooke and Paige~\cite{BP2003} and in particular contains the priority 
$\AND$ gate. 
Khand assigns truth values to his 
attack trees by giving truth tables for all gates. Khand's
truth tables, when restricted to $\AND$, $\OR$, and priority $\AND$, 
constitute an attribute domain 
which is compatible (in the sense
of~\cite{KoMaRaSc_JLC}) with the SP semantics for 
$\SAND$ attack trees as defined in this paper.
}

We observe that the extensions of attack trees with sequential conjunction are 
rather diverse in terms of application domain, interpretation, and formality. 
In order to give a clear and unambiguous interpretation of the $\SAND$ operator 
and capture different application domains, it is necessary to give a formal 
semantics as a translation to a well-understood domain. Note that, neither the 
multiset~\cite{MaOo} nor the propositional 
semantics~\cite{KoPoSc} can express ordering of attack components. 
Therefore, a richer semantical domain needs to be defined. 
The purpose of this article is to address this problem.

\section{Attack Trees with Sequential Conjunction}
\label{sec:sand-atree}

We extend the attack tree formalism so that a refinement of a
(sub-)goal of an attacker can be a sequential conjunct (denoted by
\SAND) in addition to disjuncts and conjuncts. We first give a 
definition of attack trees with the new sequential operator and then 
define series-parallel graphs on which the semantics for the new attack
trees is based.

\subsection{\SAND \ Attack Trees}\label{sec:atree-sand}

Let $\basicact$ denote the set of all possible basic actions of an attacker. 
We formalize standard attack trees introduced by Schneier in~\cite{Schn} 
and call them simply \emph{attack trees} in the rest of this
paper. Attack trees are closed terms 
over the signature $\basicact\cup \{\OR, \AND\}$, 
generated 
by the following grammar, where $b\in\basicact$ is a terminal symbol. 
\begin{equation} 
\label{gram:AT}
t ::= b \mid \OR(t, \dots, t) \mid \AND(t, \dots, t).
\end{equation}
The universe of attack trees is denoted by $\atree$. 
\emph{$\SAND$ attack trees} are closed terms over the signature
$\basicact\cup \{\OR, \AND, \SAND\}$, where 
$\SAND$ is a non-commutative operator called \emph{sequential conjunction},
and are generated by the grammar
\begin{equation} 
\label{gram:SAT}
t ::= b \mid \OR(t, \dots, t) \mid \AND(t, \dots, t) \mid \SAND(t, \dots, t
). 
\end{equation}
The universe of $\SAND$ attack trees is denoted by $\sandtree$. The 
purpose of \OR\ and \AND\ refinements in $\SAND$ attack trees
is the same as in 
attack trees. The sequential conjunctive refinement \SAND\ allows us to model 
that a certain goal is reached if and only if all 
its subgoals are reached in a precise order. 

The following attack scenario 
motivates the need for extending attack trees with sequential
conjunctive refinement. 
\begin{example}
\label{eg:def-attack-tree}
Consider 
a file server $\server$, offering
ftp, ssh, and rsh services. 
The attack tree in Figure~\ref{fig:attack-tree}
shows how an attacker can gain root privileges on 
$\server$ (\becomeroot), in two ways: 
either without providing any user credentials 
(\noauth)
or by breaching the authentication mechanism
(\auth).
\begin{figure}[t]
\centering
\includegraphics[height=4.5cm]{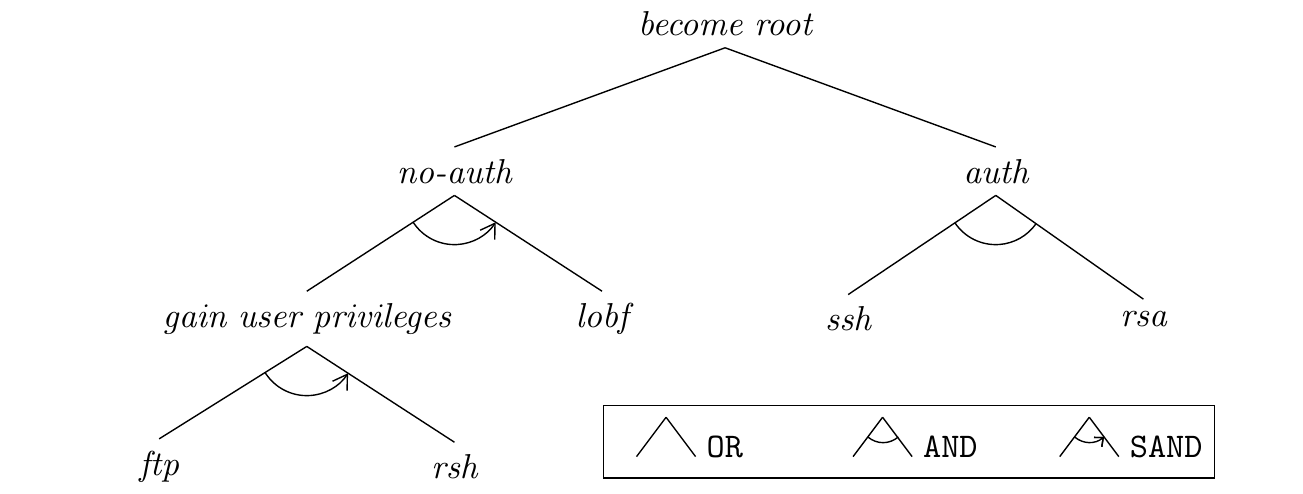}
\caption{An attack tree with sequential and parallel conjunctions}
\label{fig:attack-tree}
\end{figure}

In the first case,
the attacker must first gain user privileges on $\server$ 
(\textit{gain user pri\-vi\-leges}) 
and then perform a local buffer overflow attack (\localbofx). 
Since the attack steps must be executed in this particular order, 
the use of \SAND~refinement is substantial.
To gain user privileges, the attacker 
must exploit an ftp vulnerability 
to anonymously upload a list of trusted hosts to $\server$ 
(\ftprhostsx).\footnote{For readability, attack actions are named after
the services that are exploited.} Finally, she can use the
new trust condition to remotely execute shell commands on $\server$ (\rshx). 
 
The second way is to abuse a buffer overflow in both the 
ssh daemon (\sshbofx) and the RSAREF2 library (\rsarefbofx) used for authentication.
These attacks can be executed in any order, which is modeled with 
the standard \AND\ refinement. 

Using the term notation introduced in this section, 
we can represent the \SAND\ attack tree in Figure~\ref{fig:attack-tree} 
as
\[
t = \OR\Big(\SAND\big(\SAND(\ftprhostsx, \rshx), \localbofx\big), \\
 \AND(\sshbofx, \rsarefbofx)\Big),
\]
where $\ftprhostsx, \rshx, \localbofx, \sshbofx, \rsarefbofx\in \basicact$
are basic actions.
\end{example}

\subsection{Series-Parallel Graphs}
\label{sec:spgraphs}

A \emph{series-parallel graph} (SP graph) is an edge-labeled directed graph that has 
two unique, distinct vertices, called \emph{source} and \emph{sink}, 
and that can be constructed with
the two operators for sequential and parallel composition of graphs
that we formally define below. 
A source is a vertex which has no incoming edges and a 
sink is a vertex without outgoing edges. 

Our formal definition of SP graphs is based on \emph{multisets},
i.e., sets in which members are allowed to occur more than once. 
We use $\msl \cdot \msr$ to denote multisets and $\powerset(\cdot)$
to denote powersets. 
The \emph{support} $M^{\star}$ of a multiset $M$ is the set of 
distinct elements in $M$. 
For instance, the support of the multiset $M = \msl b_1, b_2, b_2 \msr$
is $M^{\star} = \{b_1, b_2\}$. 

In order to define SP graphs, we first introduce the notion of 
source-sink graphs labeled by the elements of $\basicact$. 

\begin{definition}
A \emph{source-sink graph} over $\basicact$ is a tuple
$G=(V,E,s,z)$, where $V$ is the set of vertices,
$E$ is a multiset of labeled edges with support $E^{\star} \subseteq V \times 
\basicact \times V$, 
$s\in V$ is the unique source, $z\in V$ is the unique sink, 
and $s\not =z$. 
\end{definition}
The sequential composition of a source-sink graph $G=(V,E,s,z)$ with a 
source-sink graph 
$G'=(V',E',s',z')$, denoted by $G\seqop G'$, is the
graph resulting from taking the 
disjoint union of $G$ and $G'$ and identifying the
sink of $G$ with the source of $G'$. 
More precisely,
let $\dot\cup$ denote the disjoint union operator
and $E^{[s/z]}$ denote the multiset of edges in $E$, where all occurrences of 
vertex $z$ are replaced by vertex $s$. 
Then we define
\[G\seqop G' = (V\setminus\{z\}\dot\cup V', E^{[s'/z]}\dot\cup E', s,
z').\]
The parallel composition, denoted by $G\parop G'$, is defined
similarly, except that the two sources are identified 
and the two sinks are identified. 
Formally, we have
\[G\parop G' = (V\setminus\{s,z\}\dot\cup V', E^{[s'/s,z'/z]}\dot\cup
E', s', z').\] 
It follows directly from the definitions that 
the sequential composition is associative and that the 
parallel composition is associative and
commutative.

We write $\xrightarrow{b}$ for the graph with a single edge labeled with 
$b$ and define SP graphs as follows. 
\begin{definition} \label{def:sp-graph}
The set $\spset$ of \emph{series-parallel graphs} (SP graphs) over 
$\basicact$ is defined inductively by the following two rules
\begin{itemize}
\item For $b\in \basicact$, $\xrightarrow{b}$ is an SP graph.
\item If $G$ and $G'$ are SP graphs, then so are 
$G\seqop G'$ and $G\parop G'$.
\end{itemize}
\end{definition}
It follows directly from Definition~\ref{def:sp-graph} 
that SP graphs are connected and acyclic. Moreover, 
every vertex of an SP graph lies on a path from the source to
the sink. 
We consider two SP graphs to be \emph{equal} 
if there is a bijection between their
sets of vertices that preserves the edges and edge labels.
{
}
\begin{example}
Figure~\ref{fig:example-spgraph} shows an example of an SP graph with
the source $s$ and the sink $z$.
This graph corresponds to the construction
\[\big(\xrightarrow{a}\parop \xrightarrow{b} \parop \xrightarrow{b}\big) \seqop
\xrightarrow{c} \seqop \Big(\big(\xrightarrow{d} \seqop (\xrightarrow{e}\parop
\xrightarrow{f})\big) \parop \xrightarrow{g}\Big).\]

\begin{figure}[t]
\centering
\includegraphics[height=2.2cm]{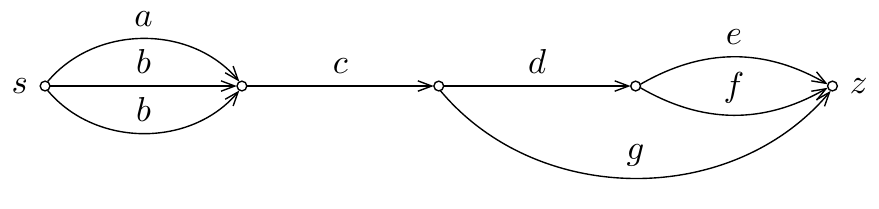}
\caption{A series-parallel graph}
\label{fig:example-spgraph}
\end{figure}

\end{example}

\subsection{SP Semantics for $\SAND$ Attack Trees}
\label{sec:atree-sand-semantics}

Numerous semantics have been proposed to 
interpret attack trees, 
including 
propositional logic~\cite{KoPoSc_iFM'14}, multisets~\cite{MaOo}, 
 De Morgan lattices~\cite{KoPoSc}, tree automata~\cite{AY2006},
and 
Markov processes~\cite{PB2010,Arnold-POST14}. 
The choice of a semantics allows us to accurately represent the
assumptions made in a security scenario, e.g., 
whether actions can be repeated or resources reused, and 
to decide which trees represent the same security scenario. 
The advantages of formalizing attack trees and the need for various
semantics have been discussed in~\cite{KoMaRaSc_JLC}.
Since attack trees are \AND/\OR~trees, the most natural interpretation 
is based on propositional logic. However, because the logical operators 
are idempotent, this interpretation assumes 
that the multiplicity of an action is irrelevant. 
As a consequence, 
the propositional semantics is not well suited to reason about 
scenarios with multiple occurrences of the same action. 
Due to this lack of expressivity a semantics was
proposed~\cite{MaOo} in which 
the multiplicity of actions is taken into account. This was achieved
by interpreting an attack tree as a set of multisets that represent
different ways of reaching the root goal. 
This multiset semantics is compatible
with computations 
that depend on the number of occurrences of an action in the tree, such as 
the minimal time to carry out the attack represented by the root goal. 
 
We now extend the multiset semantics to $\SAND$ attack trees. 
Since SP graphs naturally extend multisets with 
a partial order, they supply a formalism in which we can interpret 
trees using both --- commutative and sequential --- conjunctive refinements. 
SP graphs therefore provide a canonical semantics for \SAND\ trees in which 
multiplicity and ordering of goals and actions are significant.
The idea is to interpret an attack tree $t$ as a set of SP graphs. 
The semantics $\sem{t}=\{G_1,\dots,G_k\}$ 
of a tree $t$ corresponds to the set of possible 
attacks $G_i$, where each attack is
described by an SP graph labeled by the basic actions of $t$. 
{
}

\begin{definition} \label{def:sp-sem}
The \emph{SP semantics} for $\SAND$ attack trees is given by the function 
$\sem{\cdot}: \sandtree \to \power{\spset}$, 
which is defined recursively as follows:
for $b\in\basicact$, $t_i \in \sandtree$, $1\leq i \leq k$, 
\[
\def\arraystretch{1.3}
\begin{array}{l}
\sem{b} = \{\xrightarrow{b}\}\\
\sem{\OR(t_1, \dots, t_k)} = \bigcup_{i=1}^{k}{\sem{t_i}}\\
\sem{\AND(t_1, \dots, t_k)} =
 \{ G_1 \parop \dots \parop G_k \ \mid \ 
 (G_1, ..., G_k) \in \sem{t_1} \times ... \times \sem{t_k}
 \}\\
\sem{\SAND(t_1, \dots, t_k)} =
 \{ G_1 \seqop \dots \seqop G_k\ \mid \ 
 (G_1, ..., G_k) \in \sem{t_1} \times ... \times \sem{t_k}
 \}.\\
\end{array}
\]
\end{definition} 
{
 The SP semantics maps $\SAND$ attack trees to sets of SP graphs as follows. 
 A leaf corresponding to a basic action $b$
 is translated into a singleton set containing the SP graph 
 which consists of a single edge labeled with $b$.
 The semantics of a disjunctive node is the set 
 of all the alternative attacks described by the node's children.
 The semantics of a conjunctive node is 
 the parallel composition of every attack
 alternative from each of its children. 
 Finally, the semantics of a sequential conjunctive node is 
 a sequential composition of 
attack alternatives for the children. 
}

\begin{example}
\label{ex:sem(t)}
The SP semantics of the attack tree $t$ 
depicted in Figure~\ref{fig:attack-tree} is
\begin{equation*}
\sem{t}=\{\xrightarrow{\ftprhostsx}\xrightarrow{\rshx}
\xrightarrow{\localbofx}\ , \ 
\xrightarrow{\sshbofx} \parop \xrightarrow{\rsarefbofx}\}.
\end{equation*}
\end{example}
As shown in Example~\ref{ex:sem(t)}, the SP semantics provides an alternative 
graph representation for attack trees 
and therefore contributes
a different perspective on an attack scenario. 
The \SAND~attack tree emphasizes the refinement of goals, whereas 
SP graphs highlight the sequential aspect of attacks.
{
}

The SP semantics provides a natural partition of
$\sandtree$ into equivalence classes. 
\begin{definition}\label{def:equivalent}
Two \SAND\ attack trees $t_1$ and $t_2$ are \emph{equivalent with respect 
to the SP semantics}
if and only if they are interpreted by the same set of SP graphs, i.e., 
$\sem{t_1}=\sem{t_2}$.
\end{definition}
By Definition~\ref{def:equivalent}, if the SP semantics provides 
accurate assumptions for an attack scenario, then 
two $\SAND$ attack trees 
represent the same attack scenario 
if and only if they are equivalent 
with respect to the SP semantics. 

{We finish this section by noticing that in the case of \SAND\ attack 
trees without any 
$\SAND$ refinement, the SP semantics coincides with the multiset semantics 
introduced in~\cite{MaOo}. Indeed, it suffices to identify the multiset 
$\msl b_1, \dots, b_k \msr$ with the SP 
graph $\xrightarrow{b_1}\parop \dots \parop \xrightarrow{b_k}$. 
We discuss this issue more in details in Section~\ref{sec:SP_vs_M}.

}

\section{Axiomatization of the SP Semantics} 
\label{sec:axioms}
{
In order to provide efficient analysis methods 
for attack tree-like models, we need to be able to decide whether two 
trees are equivalent with respect to a given semantics. 
Ideally, we would like to find the most efficient (e.g., the smallest)
representation of a given security scenario. However, 
in the case of the \SAND\ semantics, 
there exists an infinite number of trees $t'$ equivalent to a given tree $t$.

In this section we study the mathematical implications of 
using sets of SP graphs as an interpretation domain for \SAND\ attack trees. 
We introduce an axiomatization of $\SAND$ attack
trees which is complete with respect to the SP semantics.
This allows us to reason directly on $\SAND$ attack trees, without having to
move to the semantical domain. Further, we derive a term
rewriting system from the axiomatization as a means to effectively
decide whether two $\SAND$ attack trees are equivalent with respect to the 
SP semantics. As a
consequence, we obtain a canonical representation of $\SAND$
attack trees which we prove to be isomorphic to sets of SP graphs.
}

\subsection{A complete set of axioms for the SP semantics}

Let $\var$ be a set of variables denoted 
by capital letters. 
Following the approach developed in~\cite{KoMaRaSc_JLC}, we 
axiomatize
$\SAND$ attack trees 
with equations 
$l=r$, where $l$ and $r$ are terms over 
variables in $\var$, 
constants in $\basicact$, and the operators \AND, \OR, and \SAND. 
The equations formalize the intended properties of refinements
and provide semantics-preserving transformations of 
\SAND\ attack trees. 
\begin{example}
\label{ex:axiom}
Let $\mathrm{Sym}_{\ell}$ denote the set of all bijections
from $\{1,\dots,{\ell}\}$ to itself. 
The axiom 
\[\AND(Y_1,\dots,Y_{\ell})=\AND(Y_{\sigma(1)},\dots,Y_{\sigma({\ell})})\] 
expresses that the order between children refining a parallel conjunctive node 
is not relevant. 
In other words, the operator \AND\ is commutative. This 
implies that any two trees of the form 
$\AND(t_1,\dots,t_l)$ and $\AND(t_{\sigma(1)},\dots,t_{\sigma(l)})$
represent the same scenario.
\end{example}

Our goal is to define a complete set of axioms, denoted by $\ESP$, for
the SP semantics for \SAND\ attack trees. Intuitively, $\ESP$
is a set of equations
that can be applied to transform 
a \SAND\ attack tree 
into any equivalent \SAND~attack tree 
with respect to the SP semantics. 
Before defining the set $\ESP$, 
we formalize the notion of a complete set of axioms for a given 
semantics for (\SAND) attack trees, following~\cite{KoMaRaSc_JLC}. 

Let $T(\var,\Sigma)$ be the free term algebra
over the set of variables
$\var$ and a signature $\Sigma$, 
and let $E$ be a set of equations over $T(\var,\Sigma)$. 
The equation $t=t'$, where $t,t'\in T(\var,\Sigma)$,
is a \emph{syntactic consequence} of $E$ 
(denoted by $E\vdash t=t'$) 
if it can be derived from $E$ by application 
of the following rules. 
For all $t,t',t''\in T(\var,\Sigma)$, $\rho\colon \var\to T(\var,\Sigma)$,
 and $X\in \var$: 
\begin{itemize}
\item $E\vdash t=t$, 
\item if $t= t' \in E$, then $E\vdash t=t'$,
\item if $E\vdash t=t'$, then $E\vdash t'=t$, 
\item if $E\vdash t=t'$ and $E\vdash t'=t''$, then $E\vdash t=t''$. 
\item if 
$E\vdash t=t'$, then
$E\vdash \rho(t)=\rho(t')$,
\item if $E\vdash t=t'$, then 
$E \vdash t'' [t/X] = t'' [t'/X]$, 
where $t'' [t/X]$ is the term obtained from $t''$
by replacing all occurrences of variable $X$ with $t$. 
\end{itemize}

 Let $\sandtreevar$ denote the set of terms constructed from the set of 
variables $\var$, the set of basic actions $\basicact$ (treated as
constants), and operators $\OR$, $\AND$ and $\SAND$.
Let $\atreevar$ be the set of terms constructed 
from the same parts, except for the operator \SAND. 
Using the notion of syntactic consequence, 
we define a complete set of axioms for a semantics for attack trees. 
\begin{definition}
\label{def/compl_set_axioms}
Let $\semantic{\cdot}$ be a semantics for attack trees 
(resp. \SAND\ attack trees)
and let $E$ be a set of equations over $\atree^{\var}$ 
(resp. $\sandtreevar$). 
The set $E$ is a \emph{complete set of axioms} for $\semantic{\cdot}$
if and only if, for all $t,t'\in \atree$ (resp. $\sandtree$)
\[\semantic{t}=\semantic{t'} \iff E\vdash t=t'. \]
\end{definition}

We are now ready to give a complete set of axioms for the SP semantics
for \SAND\ attack trees. These axioms 
allow us to determine whether two visually
distinct trees represent the same security 
scenario according to the SP semantics. 

\begin{theorem}
\label{th:axioms}
Given $k,m\geq 0$, and $\ell\geq 1$, let
$\Xvec = X_1, \ldots, X_k$,
$\Yvec = Y_1, \ldots, Y_{\ell}$, and
$\Zvec = Z_1, \ldots, Z_m$
be sequences of variables.
Let $\mathrm{Sym}_{\ell}$ be the set of all bijections
from $\{1,\dots,{\ell}\}$ to itself. The following 
set of equations over $\sandtreevar$, denoted by $\ESP$, is a complete 
set of axioms\footnote{Note that the axioms are in fact \emph{axiom schemes}. 
The operators $\OR$, $\AND$ and $\SAND$ are \emph{unranked},
representing infinitely many $k$-ary function symbols ($k\geq1$).} for 
the SP semantics for $\SAND$ attack trees.
\begingroup
\allowdisplaybreaks
\begin{align}
& \OR(Y_1,\dots,Y_{\ell})=\OR(Y_{\sigma(1)},\dots,Y_{\sigma({\ell})}),
\quad \forall \sigma\in \mathrm{Sym}_{\ell} \tag{$E_{1}$} \label{E1}\\
& \AND(Y_1,\dots,Y_{\ell})=\AND(Y_{\sigma(1)},\dots,Y_{\sigma({\ell})}),
\quad \forall \sigma\in \mathrm{Sym}_{\ell} 
\tag{$E_{2}$}\label{E2}\\ 
& \OR\big(\Xvec,\OR(\Yvec)\big)=
\OR(\Xvec,\Yvec) \tag{$E_{3}$} \label{E3}\\
& \AND\big(\Xvec,\AND(\Yvec)\big)=
\AND(\Xvec,\Yvec) \tag{$E_{4}$} \label{E4}\\
& \SAND\big(\Xvec,\SAND(\Yvec),\Zvec\big)=
\SAND(\Xvec,\Yvec,\Zvec) \tag{$E_{4'}$} \label{E4'}\\
& \OR(A)= A \tag{$E_{5}$} \label{E5}\\
& \AND(A)= A \tag{$E_{6}$} \label{E6}\\
& \SAND(A)= A \tag{$E_{6'}$} \label{E6'}\\
& \AND\big(\Xvec,\OR(\Yvec)\big)= 
\OR\big(\AND(\Xvec,Y_1),\dots,\AND(\Xvec,Y_{\ell})\big) \tag{$E_{10}$} \label{E10}\\
& \SAND\big(\Xvec,\OR(\Yvec),\Zvec\big)= 
\OR\big(\SAND(\Xvec,Y_1,\Zvec),\dots,\SAND(\Xvec,Y_{\ell},\Zvec)\big) \tag{$E_{10'}$} 
\label{E10'}\\
& \OR(A,A,\Xvec)= \OR(A,\Xvec). \tag{$E_{11}$} 
\label{E11}
\end{align}
\endgroup
\end{theorem}
The numbering of the axioms in $\ESP$ corresponds to the numbering 
of the axioms for the multiset semantics for standard attack trees, 
as presented in~\cite{KoMaRaSc_JLC}, while new axioms (involving \SAND) are 
marked with primes. 
\begin{proof}
The proof of this theorem follows the same
line of reasoning as the proofs of Theorems~4.2 and~4.3 of
Gischer~\cite{Gischer1988199}, where series--parallel pomsets are axiomatized.
{
 To prove the theorem, we remark that 
 SP graphs form a visual representation of 
 series-parallel 
 partially ordered multisets (SP pomsets). 
 A complete, finite axiomatization of pomsets under 
 concatenation, parallel composition and union
 has been provided in~\cite{Gischer1988199}, where 
 sets of series-parallel pomsets have been used to represent processes. 
 In our case, sets of series-parallel pomsets (\ie, sets of SP graphs)
 represent attack trees constructed using 
 \AND\ (having the same properties as the parallel composition of processes), 
 \SAND\ (having the same properties as concatenation),
 and \OR\ (having the same properties as choice). 
 The set $\ESP$ corresponds to the axioms from~\cite{Gischer1988199}. 
 The axioms involving the identity elements (\ie, $1$ -- the empty pomset, and 
$0$ --
 the empty process) have been omitted because they can only be used 
 for transforming processes involving $0$ or $1$ and 
 such identity elements do not exist in the case of attack trees. 
 Furthermore, our axioms are written using unranked operators 
 contrary to the binary operators of concatenation, parallel composition, and 
 choice. 
}

\qed
\end{proof}

\subsection{\SAND~Attack Trees in Canonical Form}
\label{sec:NF}
Let $\semantic{\cdot}$ be a semantics for (\SAND) attack trees.
A complete axiomatization of $\semantic{\cdot}$
can be used to derive a canonical form of trees 
interpreted with $\semantic{\cdot}$. 
Such canonical forms provide the most concise 
representation for equivalent trees
and are the natural representatives 
of equivalence classes defined by 
$\semantic{\cdot}$. 

When \SAND~attack trees are interpreted using the SP semantics, 
their canonical forms 
consist of either a single basic action, 
or of a root node labeled with \OR\ and subtrees 
with nested, alternating 
occurrences of \AND\ and \SAND\ nodes. 
Canonical forms correspond exactly to 
the sets of SP graphs labeled by $\basicact$
and they depict all attack alternatives in a straightforward 
way.

{
Canonical representations of \SAND~attack trees
under the SP semantics
can be defined using the complete set of axioms 
$\ESP$. 
By orienting the equations 
\eqref{E3}, \eqref{E4}, \eqref{E4'}, \eqref{E5}, \eqref{E6},
\eqref{E6'}, \eqref{E10}, \eqref{E10'}, and \eqref{E11}
from left to right, we obtain a term rewriting system, 
denoted by $\thetrs$. 
The canonical representations of \SAND~attack trees 
correspond to normal forms with respect to $\thetrs$. 
In the rest of this section we show that 
the normal forms with respect to $\thetrs$
are exactly the terms generated by the following grammar, 
where $k\geq 2$ and $b\in\basicact$
\begin{align} 
\label{grammar:NF}
\nfterm ::= \quad& \asterm \mid \OR(\asterm_1, \dots, \asterm_k)
 \quad\mbox{for } \asterm_i \neq \asterm_j
 \mbox{ if } i\neq j\\\notag
\asterm ::= \quad & \andterm \mid \sandterm\\\notag
\andterm ::= \quad & b \mid \AND(\sandterm_1, \ldots, \sandterm_k)\\\notag
\sandterm ::= \quad& b \mid \SAND(\andterm_1, \ldots, \andterm_k).\\\notag
\end{align}
The non-terminal $\andterm$ produces all trees that consist of
a single basic action or being 
a nested alternation of 
$\AND$ and $\SAND$ operators,
where the outer operator is $\AND$.
Similarly, $\sandterm$ produces all such trees 
where the outer operator is $\SAND$.
The non-terminal $\asterm$ generates the two 
previously described types of trees. 
Finally, $\nfterm$ 
combines the trees generated by $\asterm$ using 
the \OR\ refinement. 
We denote the sets of terms 
generated by $\nfterm$, $\asterm$, $\andterm$, and $\sandterm$, by $\nfterms$, 
$\asterms$, $\andterms$, and $\sandterms$, respectively. 

We first observe that the terms generated by the non-terminal $\nfterm$
correspond exactly to all sets of SP graphs labeled by the elements of 
$\basicact$. 
\begin{lemma}\label{lem:bijection}
The restriction of function $\sem{.}$ to $\nfterms$
is a bijection from $\nfterms$ to $\power{\spset}$.
\end{lemma}
\begin{proof}
The proof consists of two steps. First we prove that the terms from
$\asterms$ exactly correspond to SP graphs, after which we extend this
result to the correspondence between $\nfterms$ and the sets of SP graphs.
\begin{enumerate}
\item $\sem{.}$ is a bijection from $\asterms$ to $\spset$.\\
 For injectivity we prove by induction  that $\sem{\asterm_1} = \sem{\asterm_2}$
 implies $\asterm_1 = \asterm_2$. Given that for every $t \in \asterms
 $ the set $\sem{t}$ contains a single SP graph, we abuse notation and 
 refer to 
 $\sem{t}$ as the SP graph contained in $\sem{t}$. Let us assume that 
 $\asterm_1$ is a basic action $b$, then $\sem{\asterm_1}$ is a single 
 edge graph $\xrightarrow{b}$, it follows that $\asterm_2 = 
 b$ considering the edge labels preservation property of two isomorphic SP 
 graphs. Let us consider now that $\asterm_1 = \AND(\sandterm_1^1, \cdots, \sandterm_k^1)$ for 
 some $k \geq 2$, which means that $\sem{\asterm_1} = 
 \sem{\sandterm_1^1} \parop \cdots \parop
 \sem{\sandterm_k^1}$. Given that $\sem{\asterm_1} = \sem{\asterm_2}$,
 for every $i \in \{1, \cdots, k\}$, $\sem{\asterm_2}$ contains a subgraph 
 $G_i$ which is isomorphic to $\sem{\sandterm_i^1}$. Thus, considering that 
 $\sem{\sandterm_i^1}$ is either a basic action or a sequential composition, 
 we have that 
 $\sem{\asterm_2} = G_1 \parop \cdots \parop G_k$ and, according to grammar
 \eqref{grammar:NF}, 
 $\asterm_2 = \AND(\sandterm_1^2, \cdots, \sandterm_k^2)$, for some terms 
 $\sandterm_1^2, \cdots, \sandterm_k^2 \in \sandterms$. By the 
 induction hypothesis, it 
 follows that $\sem{\sandterm_i^1} = G_i = 
 \sem{\sandterm_i^2}$ implies $\sandterm_i^1 = \sandterm_i^2$, which leads to 
 $\asterm_1 = \asterm_2$. A similar 
 proof can be obtained for the case where $\asterm_1 = \SAND(\andterm_1^1, 
 \cdots, \andterm_k^1)$.

 Surjectivity follows from the fact that every SP graph has a unique
 (modulo associativity) decomposition in terms of 
 the operators for 
 sequential and parallel composition.
 Such a decomposition naturally corresponds to terms from $\asterms$.
\item $\sem{.}$ is a bijection from $\nfterms$ to $\powerset(\spset)$.\\
 For injectivity, we assume that $\sem{\nfterm_1} = \sem{\nfterm_2}$.
 This implies that the sets $\sem{\nfterm_1}$ and $\sem{\nfterm_2}$
 have the same size and the same elements. If this size is $1$, then
 they both contain the same element, which uniquely corresponds to
 a term $\asterm$, so $\nfterm_1 = \asterm = \nfterm_2$.
 If the size of the sets is larger than $1$, then $\nfterm_i$ 
 are of the form $\OR(\asterm_1^i,\dots, \asterm_k^i)$, for 
 $i\in\{1,2\}$. Since $\sem{\nfterm_1} = \sem{\nfterm_2}$, 
 the elements of these two
 sets are pairwise identical.
 Moreover, by definition, 
 all arguments in $\nfterm_i$ are different, which implies 
 $k_1=k_2$.
 From the previous item, it follows
 that the elements of $\sem{\nfterm_i}$ correspond uniquely 
 to $\asterm_1^i,\dots, \asterm_k^i$. 
 From the pairwise equality between the
 arguments of the two terms, it follows that $\nfterm_1$ and
 $\nfterm_2$ are identical.
 
 For surjectivity, 
 let $\{G_1, \dots, G_k\}$ 
 be a set of SP graphs. 
 It follows from the previous item
 that there exist trees $\asterm_1, \dots, \asterm_k$, 
 such that $\sem{\asterm_i}=G_i$, for $i\in\{1,\dots, k\}$. 
 This implies that $\{G_1, \dots, G_k\}= \sem{\OR(\asterm_1, \dots, \asterm_k)}$, 
 which finishes the proof of surjectivity. 
\end{enumerate}
\qed
\end{proof}
Lemma~\ref{lem:bijection} shows that 
the grammar \eqref{grammar:NF} 
generates all \SAND~attack trees in canonical form. 
It remains to be proven that
the set of trees generated by the grammar \eqref{grammar:NF}
is equal to the set of normal forms of the term rewriting system
$\thetrs$ and that these normal forms are unique. 

\begin{theorem}
\label{th:trs}
The term rewriting system $\thetrs$ is strongly terminating and confluent.
\end{theorem}
\begin{proof}
We show that the term rewriting system $\thetrs$ is terminating and confluent
with help of the grammar~\eqref{grammar:NF}, 
in four steps. 
\begin{enumerate}
\item 
First, we show with
standard methodology that the
term rewriting system is terminating.
We define the following norm which assigns natural numbers to terms:
\[
\begin{array}{lll}
\norm{b} & = & 1\\
\norm{\OR(X_1, \ldots, X_k)} & = &
 \norm{X_1} + \ldots + \norm{X_k} + 2\\
\norm{\AND(X_1, \ldots, X_k)} & = &
 2\cdot \norm{X_1} \cdot \ldots \cdot \norm{X_k}\\
\norm{\SAND(X_1, \ldots, X_k)} & = &
 2\cdot \norm{X_1} \cdot \ldots \cdot \norm{X_k}\\
\end{array}
\]

It can be easily verified that for every rewrite rule $l\rightarrow r\in \thetrs$,
we have $\norm{l} > \norm{r}$. Consequently, there are no
infinite reduction sequences, or, in other words, the term rewriting
system is strongly terminating.
Notice that because we consider term rewriting modulo commutativity of
$\OR$ and $\AND$, we have to verify that the left-hand side and the right-hand
side of equations \eqref{E1} and \eqref{E2}
have equal norms~\cite{KBV2001}. This is clearly the
case.
\item 
Now we prove that the terms 
produced by the grammar~\eqref{grammar:NF} 
are exactly the normal forms with respect to $\thetrs$.
For the terms in $\asterms$, none of the rewrite rules can be applied,
because these terms do not contain $\OR$, have no occurrences of
$\AND$ containing an argument of type $\AND$, have no occurrences of
$\SAND$ containing an argument of type $\SAND$, and do not contain
operators with a single argument.
We extend this to terms $\nfterms$ by observing that all $\OR$
operators occurring in such terms have at least two arguments and that
all these arguments are different.

Conversely, consider a term $t$ in normal form that 
contains an $\OR$ operator. Then $t=\OR(t_1,\ldots,t_n)$, where the $t_i$
do not contain an $\OR$ operator, else \eqref{E3}, \eqref{E10}, or
\eqref{E10'} can be applied. 
It remains to show that normal form terms without occurrence of an
$\OR$ operator are in $\asterms$. Such terms are basic terms or have
$\SAND$ or $\AND$ as their top-level operator. The last two cases are
symmetric and we therefore only consider the case
$\AND(t_1,\ldots,t_n)$. We must show that each $t_i$ is a basic term
or in the form $t_i=\SAND(t_1',\ldots,t_m')$. Suppose not, then there
exists a $t_i$ that has $\AND$ as its top-level operator. It follows
that the term is not in normal form because \eqref{E4} can be applied.

\item The normal forms are unique. 
To show that the normal forms are unique, assume that 
$\nfterm_1$ and $\nfterm_2$ are both normal forms for a
 \SAND~attack tree $t$. 
Since the rewrite system $\thetrs$ 
was constructed by orienting the axioms from $\ESP$, 
we have that $\ESP\vdash \nfterm_1=\nfterm_2$. 
This means that $\sem{\nfterm_1}=\sem{\nfterm_2}$.
From bijectivity proven in Lemma~\ref{lem:bijection}, we obtain 
$\nfterm_1 = \nfterm_2$.

\item Now that we have proven termination and uniqueness of normal forms, it
immediately follows that the term rewriting system is
confluent~\cite{D2005}.
\end{enumerate} 
\qed
\end{proof}

Example~\ref{ex:nf} illustrates the notion of 
canonical form for \SAND~attack trees.
}

\begin{example}
\label{ex:nf}
The canonical form of the \SAND~attack tree $t$ 
in Figure~\ref{fig:attack-tree} is the tree 
\[
t' = \OR\Big(\SAND\big(\ftprhostsx, \rshx, \localbofx\big), \\
 \AND\big(\sshbofx, \rsarefbofx\big)\Big)
\]
shown in Figure~\ref{fig:attack-tree2}. 
It 
is easily seen to be in normal form with respect to $\thetrs$.
\end{example}
\begin{figure}[h]
\centering
\includegraphics[height=3.4cm]{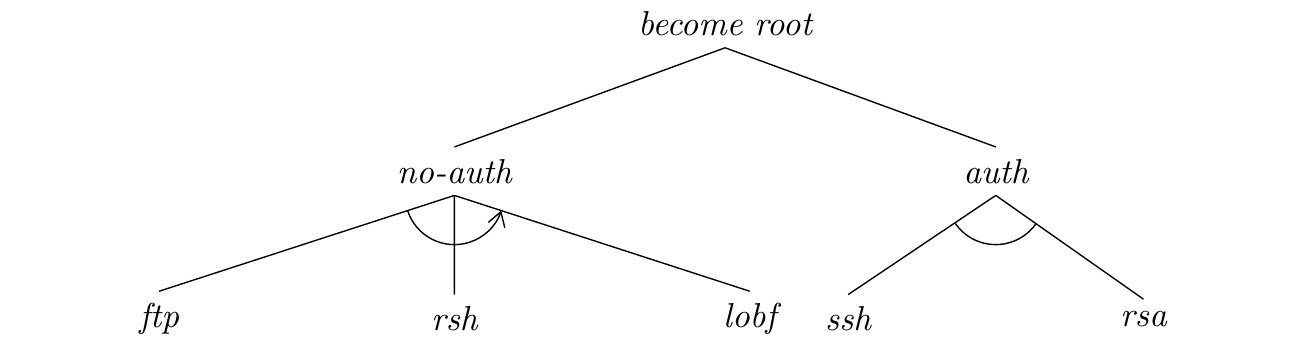}
\caption{\SAND\ attack tree $t'$ equivalent to \SAND\ attack tree $t$ from 
Figure~\ref{fig:attack-tree}}
\label{fig:attack-tree2}
\end{figure}

\subsection{SP Semantics as a Generalization of the Multiset Semantics}
\label{sec:SP_vs_M}

Having a complete set of axioms for the SP semantics
allows us to formalize the relation between 
\SAND\ attack trees under the SP semantics and 
attack trees under the multiset semantics, denoted by $\msem{\cdot}$. 
This is achieved by extracting 
a complete set of axioms for 
the multiset semantics for 
attack trees from 
the set $\ESP$. 
Let $\ESM$ be the subset of axioms from $\ESP$ that
do not contain the $\SAND$ operator, \ie, 
$\ESM=\{$\eqref{E1}, \eqref{E2}, \eqref{E3}, \eqref{E4}, \eqref{E5},
\eqref{E6}, \eqref{E10}, \eqref{E11}$\}$.
\begin{theorem}
\label{th:multiset-axioms}
The axiom system $\ESM$
is a complete set of axioms for the multiset
semantics for attack trees.
\end{theorem}
{
\begin{proof}
In \cite[Theorem 4.9]{KoMaRaSc_JLC}, a complete axiomatization of
the multiset semantics for an extention of attack trees 
called attack--defense trees (ADTrees) is given. In the following, we
call that axiomatization $\EADT$.
ADTrees are a superset of attack trees.
They may contain defender's nodes 
modeled by the so called
opponent's functions and countermeasures. 
We claim that $\ESM$
is a complete axiomatization of the multiset semantics for attack trees.
Obviously if two attack trees
are equal with respect to $\ESM$, then they
are also equal with respect to $\EADT$.
This is clear, because $\ESM\subset \EADT$.

Conversely, we prove that if two attack trees are equal with respect to 
$\EADT$, they are equal with respect to $\ESM$. This
follows from the following syntactical reasoning. 
$\EADT$ contains function symbols which we call
\emph{countermeasures}. 
Observe by inspecting the axioms of $\EADT$ that if a
countermeasure occurs at the left-hand side of an equation, then it also
occurs at the right-hand side, and vice versa. Therefore, axioms 
$(E_{13}),
(E_{16}), (E_{17}), (E_{18}), (E_{19}), (E_{20})$ from $\EADT$
can never be used in a derivation of equality of
two standard attack trees. Further, observe that the remaining 
axioms $(E_9)$ and $(E_{12})$ from $\EADT$ make use of 
\emph{opponent's functions}.
In these axioms, 
an opponent function occurs on the left-hand side if and only
if it occurs on the right hand side. Thus 
these axioms are never used to equate
two attack trees which do not contain opponent's nodes. 
The remaining axioms are precisely 
\eqref{E1}, \eqref{E2}, \eqref{E3}, \eqref{E4}, \eqref{E5},
\eqref{E6}, \eqref{E10}, \eqref{E11}.
So, we can only use these axioms 
to derive equalities
of attack trees with respect to $\EADT$, which implies that such a
derivation is also possible using axioms from $\ESM$.
\qed
\end{proof}
}

By comparing the complete sets of axioms 
$\ESP$ and $\ESM$ we obtain that two attack trees are equivalent under
the multiset semantics if and only if they are equivalent under the SP 
semantics. This is formalized in the following theorem. 
\begin{theorem}
\label{th:SP_vs_M}
$\SAND$ attack trees under the SP semantics are a \emph{conservative
extension} of attack trees under the multiset semantics.
\end{theorem}
\begin{proof}
{
Let $t$ and $t'$ be standard attack trees.
Let $\msem{t}$ and $\msem{t'}$ be their interpretation in the 
multiset semantics and $\sem{t}$ and
$\sem{t'}$ be their interpretation in the SP semantics.
We prove that $\msem{t} = \msem{t'}$ if and only if $\sem{t} =
\sem{t'}$.

By Theorem~\ref{th:multiset-axioms}, a complete axiomatization of 
the multiset semantics for attack 
trees 
consists of axioms~\eqref{E1}, \eqref{E2}, \eqref{E3}, \eqref
{E4}, \eqref{E5}, \eqref{E6}, \eqref{E10}, \eqref{E11}.
The complete axiomatization of the SP semantic for 
$\SAND$ attack trees additionally 
contains axioms \eqref{E4'}, \eqref{E6'},
and \eqref{E10'}. Thus, every equivalence of attack
trees under the multiset semantics 
is clearly an equivalence of $\SAND$ attack trees
under the SP semantics.

To see the converse, 
we show that the additional axioms do not introduce new
equalities on standard attack trees. 
First inspect the three additional axioms and note that all of them
contain the $\SAND$ operator.

Next, observe that for all axioms, the set of variables occurring on the
left-hand side is equal to the set of variables occurring on the
right-hand side. Thus, there is no axiom eliminating all
occurrences of a variable. 
In particular, we claim that 
all axioms transform 
terms containing a $p$-ary $\SAND$ expression, where $p\geq 2$, into
terms containing a $q$-ary $\SAND$ expression, for some $q\geq 2$. 
This is evident for equations without the $\SAND$ operator (since no
variables are eliminated) and remains
to be shown for equations \eqref{E4'}, \eqref{E6'},
and \eqref{E10'}. 
Axiom~\eqref{E6'} introduces and removes unary $\SAND$, but does not
modify the single variable $A$ and therefore satisfies the claim. 
The arities of the two left-hand side $\SAND$ operators 
in equation~\eqref{E4'} are $l$ and $k+l+m$ and the arity of the 
right-hand side operator is $k+l+m$, where $k,m\geq 0$ and $l\geq 1$.
Since $1\leq l\leq k+l+m$ and both sides contain a $\SAND$ operator of
arity $k+l+m$, if either of the two sides
contains a $\SAND$ operator with two or more arguments, then so does
the other side.
Finally, since $l\geq 1$, 
the arity of the $\SAND$ operator on the left-hand side of
equation~\eqref{E10'} is equal to the arities of the $\SAND$ operators
on its right-hand side and at least one $\SAND$ operator
occurs on the right-hand side.

We can now show that none of the three axioms~\eqref{E4'}, \eqref{E6'}, \eqref
{E10'} introduces new equalities on
standard attack trees. 
In particular, axiom~\eqref{E6'} introduces and removes unary
$\SAND$, but this does not introduce new equalities on standard attack
trees. Equations~\eqref{E4'} and~\eqref{E10'} match unary $\SAND$, but
require a further $\SAND$ with $2$ or more arguments to add a new equality.
Since, by the above claim, no $p$-ary $\SAND$ for $p\geq 2$ 
can be introduced with any of
the equations, the additional equations do not introduce new
equalities on standard attack trees.
\qed
}

\end{proof}

{
}

\section{Attributes}
\label{sec:attributes}

Attack trees do not only serve to represent security scenarios 
in a graphical way. They can also be used to 
quantify such scenarios with respect to 
a given parameter, called an \emph{attribute}. 
Typical examples of attributes include
the likelihood that the attacker's goal is satisfied and 
the minimal time or cost of an attack. 
Schneier described~\cite{Schn} 
an intuitive bottom-up algorithm for calculating attribute values
on attack trees:
attribute values are assigned to the leaf nodes 
and two functions\footnote{These are actually 
families of functions representing infinitely many $k$-ary function symbols,
for all $k \geq 2$.} (one for the \OR\ and one for the \AND\ refinement)
are used to propagate the attribute value up to the root node. 
Mauw and Oostdijk showed~\cite{MaOo} that 
if the binary operations induced by the two functions 
define a semiring, 
then the 
evaluation of the attribute on two attack trees 
equivalent with respect to the multiset semantics 
yields the same value. This result has 
been generalized to any semantics and attribute that satisfy a notion of 
\emph{compatibility}~\cite{KoMaRaSc_JLC}. 
We briefly discuss it for \SAND~attack trees 
at the end of this section. 
We start with a demonstration on how 
the bottom-up evaluation 
algorithm 
can naturally be extended to \SAND\ attack trees.

An {\em attribute domain for an attribute $\fullattr$ on \SAND\ attack trees} 
is a tuple $\attrdomain_\attr = (\attrval_\attr, \attror_\attr, \attrand_\attr, \attrsand_\attr)$
where $\attrval_\attr$ is a set of values and 
 $\attror_\attr, \attrand_\attr, \attrsand_\attr$ are families of 
$k$-ary functions of the
 form $\attrval_\attr\times\dots\times\attrval_\attr\to\attrval_\attr$, 
associated to \OR, \AND, and \SAND\ refinements, respectively. 
An {\em attribute for \SAND~attack trees} is a pair 
$\fullattr = (\attrdomain_\attr, \basicassign_\attr)$ 
formed by an attribute domain $\attrdomain_\attr$ and a function 
$\basicassign_\attr:\basicact\to\attrval_\attr$,
called {\em basic assignment} for $\fullattr$, which associates a value from 
$\attrval_\attr$ with each basic action $b \in \basicact$. 
\begin{definition}
\label{def:attr}
Let $\fullattr = \big((\attrval_\attr, \attror_\attr, \attrand_\attr, \attrsand_
\attr), \basicassign_\attr\big)$
be an attribute. 
The attribute evaluation function 
$\attr: \sandtree \to \attrval_\attr$ which calculates the value of 
attribute 
$\fullattr$ 
for every \SAND \ attack tree $t \in \sandtree$
is defined recursively as follows
\[ \attr(t) = \left\{ 
 \begin{array}{l l}
 \basicassign_\attr(t) & \quad \text{if $t=b,\ b\in\basicact$}\\
 \attror_\attr\big(\attr(t_1), \dots, \attr(t_k)\big) & \quad \text{if $t=\OR(t_1, \dots, t_k)$}\\
 \attrand_\attr\big(\attr(t_1), \dots, \attr(t_k)\big) & \quad \text{if $t=\AND(t_1, \dots, t_k)$}\\
 \attrsand_\attr\big(\attr(t_1), \dots, \attr(t_k)\big) & \quad \text{if $t=\SAND(t_1, \dots, t_k)$}
 \end{array} \right.\]
\end{definition}

The following example illustrates the bottom-up evaluation
of the attribute \emph{minimal attack time} on the \SAND~attack trees 
given in Example~\ref{eg:def-attack-tree}.
\begin{example}
\label{ex:min_cost}
Let $\attr$ denote the minimal time 
that the attacker needs to achieve her goal.
We make the following assignments to the basic actions: 
$\ftprhostsx\mapsto 3$, $\rshx\mapsto 5$, $\localbofx\mapsto 7$, $\sshbofx\mapsto 
8$, $\rsarefbofx\mapsto 9$. 
Since we are interested in the minimal attack time, 
the function for an \OR\ node 
is defined by 
$\attror_\attr(x_1,\dots,x_k)=\min\{x_1,\dots, x_k\}$. 
The function for an \AND\ node 
is $\attrand_\attr(x_1,\dots,x_k)=\max\{x_1,\dots, x_k\}$, 
which models that the children 
of a conjunctively refined node are executed in parallel.
Finally, in order to model that the children 
of a \SAND\ node need to be executed 
sequentially, 
we let 
$\attrsand_\attr(x_1,\dots,x_k)=\sum_{i=1}^{k} x_i$. 
According to Definition~\ref{def:attr}, 
the minimal attack time for our running scenario $t$ is
\[\attror_\attr\Big(\attrsand_\attr\big(\attrsand_\attr(3,5),7\big),
\attrand_\attr(8, 9)\Big)=
\min\Big({\mathrm\Sigma}\big({\mathrm\Sigma}(3,5),7\big), \max(8, 9)\Big)=9.\]
\end{example}

In the case of standard attack trees, 
the bottom-up procedure uses only two functions to propagate 
the attribute values to the root -- one for conjunctive and one for 
disjunctive nodes. This means that the same function is employed 
to calculate the value of every conjunctively refined node, 
independently of whether its children need to be executed sequentially or can 
be executed simultaneously. 
Evidently, with \SAND~attack trees, we can apply different propagation 
functions for \AND\ and \SAND\ nodes, as in 
Example~\ref{ex:min_cost}. 
Therefore, \SAND~attack trees can be evaluated over a larger set of 
attributes, and hence may provide more accurate evaluations of 
attack 
scenarios than 
standard attack trees.

To guarantee that the evaluation of an attribute on equivalent 
attack trees yields
the same value, 
the attribute domain must be 
\emph{compatible} with a considered semantics~\cite{KoMaRaSc_JLC}. 
Our complete set of axioms is a useful tool 
to check for compatibility. 
Consider an attribute domain
$\attrdomain_\attr = (\attrval_\attr, \attror_\attr, \attrand_\attr, \attrsand_\attr)$, 
and let $\sigma$ be a mapping
$\sigma=\{\OR\mapsto \attror_\attr, \AND \mapsto \attrand_\attr, 
\SAND \mapsto \attrsand_\attr\}$. 
Guaranteeing that $\attrdomain_\attr$ is compatible with 
a semantics axiomatized by $E$ 
amounts to verifying that the equality 
$\sigma(l)=\sigma(r)$ holds in $\attrval_\attr$, for every axiom
$l=r\in E$.
It is an easy exercise to show that the attribute domain 
for minimal attack time, considered in Example~\ref{ex:min_cost},
is compatible with the SP semantics for \SAND\ attack trees.

\section{Conclusions}
\label{sec:conclusions}

We have formalized the extension of attack trees 
with sequential conjunctive refinement, 
called \SAND, 
and given a semantics to \SAND~attack trees in terms of
sets of series-parallel graphs.
This SP semantics 
naturally extends the multiset semantics 
for attack trees from~\cite{MaOo}.
We have shown that the notion of 
a complete set of axioms for a semantics and 
the bottom-up evaluation procedure can be generalized 
from attack trees to \SAND~attack trees, 
and have proposed a complete axiomatization of the SP semantics.

A number of recently proposed solutions
focus on extending attack trees with defensive
measures~\cite{RoKiTr2,KoMaRaSc_JLC}.
These extensions support reasoning about security scenarios
involving two players -- an attacker and a defender --
and the interaction between them.
In future work, we intend to add
the $\SAND$ refinement to such trees.
Afterwards, we plan to investigate
sequential disjunctive refinement, as used
for instance in~\cite{Arnold-POST14}. Our goal is to propose
a complete formalization
of trees with attack and defense nodes,
that have parallel and sequential, conjunctive and disjunctive
refinements. 
The findings will be implemented in the software 
application ADTool~\cite{adtool}.

\paragraph{\textbf{Acknowledgments}} The research leading
to these results has received funding from the European Union Seventh Framework
Programme under grant agreement number 318003 (TREsPASS) and from
the Fonds National de la Recherche Luxembourg under grant 
C13/IS/5809105.

\end{document}